\documentclass[11pt,reqno]{amsart}
\usepackage{amssymb}
\usepackage{stmaryrd}
\usepackage{stmaryrd}
\usepackage{graphics}
\usepackage{xcolor}
\usepackage{enumerate}
\usepackage[pagebackref, colorlinks = true, linkcolor = blue, urlcolor  = blue, citecolor = red]{hyperref}
\usepackage[margin=1in]{geometry}
\usepackage{calc}

\newtheorem{theorem}{Theorem}[section]
\newtheorem{definition}[theorem]{Definition}
\newtheorem{proposition}[theorem]{Proposition}
\newtheorem{corollary}[theorem]{Corollary}
\newtheorem{lemma}[theorem]{Lemma}

\begin{document}

\title[Separability of unitarily invariant random quantum states]{On the separability of unitarily invariant random quantum states -- the unbalanced regime}
\date{\today}

\author{Ion Nechita}
\address{CNRS, Laboratoire de Physique Th\'{e}orique, IRSAMC, Universit\'{e} de Toulouse, UPS, F-31062 Toulouse, France}
\email{nechita@irsamc.ups-tlse.fr}

\subjclass[2000]{}

\begin{abstract}
We study entanglement-related properties of random quantum states which are unitarily invariant, in the sense that their distribution is left unchanged by conjugation with arbitrary unitary operators. In the large matrix size limit, the distribution of these random quantum states is characterized by their limiting spectrum, a compactly supported probability distribution. We prove several results characterizing entanglement and the PPT property of random bipartite unitarily invariant quantum states in terms of the limiting spectral distribution, in the unbalanced asymptotical regime where one of the two subsystems is fixed, while the other one grows in size.  
\end{abstract}

\maketitle

\tableofcontents

\section{Introduction}

In Quantum Information Theory, when one needs to understand properties of \emph{typical} density matrices, it is necessary to endow the convex body of quantum states with a natural, physically motivated probability measure, in order to compute statistics of the relevant quantities. Since the late 1990's, there have been several candidates for such measures: the induced measures \cite{zyczkowski2001induced}, the Bures measure \cite{hall1998random}, or random matrix product states \cite{garnerone2010typicality}, just to name a few. 

The induced measures of {\.Z}yczkowski and Sommers have received the most attention, mainly due to their simplicity and to their natural physical interpretation: a density matrix from the induced ensemble is obtained by tracing an environment system of appropriate dimension out of a random uniform bipartite pure state (the latter being distributed along the Lebesgue measure on the unit sphere of the corresponding complex Hilbert space). 

In \cite{aubrun2012partial}, Aubrun studied bipartite random quantum states from the induced ensemble, and determined for which values of the ratio environment size/system size the random states are, with high probability, PPT (i.e.~they have a positive semidefinite partial transpose). Aubrun's idea was developed and generalized in many directions, for other entanglement-related properties and in different asymptotic regimes in the following years \cite{aubrun2012realigning,fukuda2013partial,banica2013asymptotic,jivulescu2014reduction,jivulescu2015thresholds,banica2015block,lancien2016k,puchala2016distinguishability}. One of the most notable results in this framework is the characterization of the \emph{entanglement threshold} from \cite{aubrun2014entanglement}, in which the authors determine, up to logarithmic factors, how large should the environment be in order for a random bipartite quantum state from the induced ensemble to be separable.

In this work, we consider random quantum states which have the property that their distribution is left unchanged by conjugation with arbitrary unitary operations; we call them \emph{unitarily invariant}. These distributions are characterized only by their spectrum, and we consider sequences of distributions with the property that their spectra converge towards some compactly supported probability measure $\mu$ on the real line. In particular, the family of distributions we consider generalizes the induced ensemble, which corresponds to a Mar{\v{c}}enko-Pastur limiting spectral distribution. We provide conditions such that the quantum state corresponding to a random unitarily invariant matrix will be, with large probability, PPT,  separable, or entangled. We shall ask that the conditions be simple, and only depend on the asymptotic spectrum of the random matrices. We state now an informal version of some of the main results contained in this paper; we refer the reader to Theorem \ref{thm:bounds-mu-Gamma} and Propositions \ref{prop:bounds-mu-Delta}, \ref{prop:bounds-X} for the exact results. 

\begin{theorem}
Let $X_d \in \mathcal M_{n}(\mathbb C) \otimes \mathcal M_{d}(\mathbb C)$ a sequence of unitarily invariant random matrices converging ``strongly'' to a compactly supported probability measure $\mu$; here, $n$ and $\mu$ are fixed. Assume that the limiting spectral measure $\mu$ has average $m$, variance $\sigma^2$, and is supported on the interval $[A,B] \subseteq [0,\infty)$. Then, 
\begin{itemize}
\item If the following condition holds, then the sequence $(X_d)_d$ is asymptotically PPT:
$$n(m-2\sigma) > B-A + 2\sigma.$$
\item If one of the two following conditions holds, then the sequence $(X_d)_d$ is asymptocally separable:
\begin{align*}
(n^2+n-1)A &> B + m(n^2-2) + 2\sigma \sqrt{n^2-2}\\
A & >(n^2-2)(B-m) + 2\sigma\sqrt{n^2-2}.
\end{align*}
\item If the following condition holds, then the sequence $(X_d)_d$ is asymptotically entangled:
$$\frac B m < 1 + n \frac{\sigma^2}{m^2} - 2 \frac{\sigma}{m}\sqrt{n-1}.$$
\end{itemize}
\end{theorem}

The paper is organized as follows: Sections \ref{sec:unitarily-invariant}, \ref{sec:free-prob}, \ref{sec:entanglement} contain facts from the theories of, respectively, unitarily invariant random matrices, free probability, and entanglement, which are used later in the paper. Section \ref{sec:block-modified-strong} contains a strengthening of a result about block-modifications of random matrices which allows us to study the behavior of the extremal eigenvalues of such matrices. Sections \ref{sec:partial-transposition}, \ref{sec:sufficient}, \ref{sec:necessary} contain the new results of this work, spectral conditions that unitarily invariant random matrices must satisfy in order to, respectively, have the PPT property, to be separable, or to be entangled. Moreover, Section \ref{sec:necessary} contains results about the asymptotic value the $S(k)$ norms introduced by Johnston and Kribs take on unitarily invariant random matrices. Finally, in Section  \ref{sec:SN-PPT}, we show that shifted GUE matrices are PPT and have Schmidt number that scales linearly with the dimension of the fixed subsystem in the unbalanced asymptotical regime.

\medskip

\noindent {\it Acknowledgments.} This research has been supported by the ANR projects {StoQ} {ANR-14-CE25-0003-01} and {NEXT} {ANR-10-LABX-0037-NEXT}, and by the PHC Sakura program (project number: 38615VA). The author acknowledges the hospitality of the TU M\"unchen, where part of this work was conducted.

\section{Unitarily invariant random matrices and strong convergence}\label{sec:unitarily-invariant}

In this work, we shall be concerned with \emph{unitarily invariant} random matrices: these are self-adjoint random matrices $X \in \mathcal M_d^{sa}(\mathbb C)$ having the property that for any unitary matrix $U \in \mathcal U_d$, the random variables $X$ and $UXU^*$ have the same distribution. From the invariance of the Haar distribution on $\mathcal U_d$, it follows that, given a deterministic matrix $A \in \mathcal M_d^{sa}(\mathbb C)$ and a Haar-distributed random unitary matrix $U \in \mathcal U_d$, the distribution of the random matrix $X:=UAU^*$ is unitarily invariant; this is the most common construction of unitarily invariant ensembles. 

The most well-studied ensembles of random matrices are, without a doubt, \emph{Wigner ensembles}  \cite{wigner1955characteristic}: these are random matrices $X \in \mathcal M_d^{sa}(\mathbb C)$ having independent and identicallly distributed (i.i.d.) entries, up to the symmetry condition $X_{ji} = \bar{X}_{ij}$, see \cite[Section 2]{anderson2010introduction}. At the intersection of Wigner and unitarily invariant ensembles is the \emph{Gaussian unitary ensemble} (GUE). A random matrix $X \in \mathcal M_d^{sa}(\mathbb C)$ is said to have $\mathrm{GUE_d}$ distribution if its entries are as follows:
$$X_{jk} = \begin{cases}
\frac{1}{\sqrt d} A_{jk}&\qquad \text{ if } j=k\\
\frac{1}{\sqrt{2d}} (A_{jk} + i B_{jk})&\qquad \text{ if } j<k\\
\frac{1}{\sqrt{2d}} (A_{kj} - i B_{kj})&\qquad \text{ if } j>k\\
\end{cases}$$
where $A_{jk},B_{jk}$ are i.i.d. real, centered standard Gaussian random variables. 

The celebrated Wigner theorem states that GUE random matrices converge in moments, as $d \to \infty$ towards the semicircle law. 

\begin{theorem}\label{thm:Wigner}
Let $X_d$ be a sequence of GUE random matrices. Then, for all moment orders $p \geq 1$, we have
$$\lim_{d \to \infty} \mathbb E \frac 1 d \operatorname{Tr} X_d^p = \int x^p d\mathrm{SC}_{0,1}(x) = \begin{cases}
\mathrm{Cat}_{p/2}&\qquad \text{ if $p$ is even},\\
0&\qquad \text{ if $p$ is odd},
\end{cases}$$
where $\mathrm{Cat}_p$ are the Catalan numbers and $\mathrm{SC}_{a,\sigma}$ is the semicircular distribution with mean $a$ and variance $\sigma^2$:
$$\mathrm{SC}_{a,\sigma} = \frac{\sqrt{4\sigma^2-(x-a)^2}}{2\pi \sigma^2} \mathbf{1}_{[a-2\sigma ,a+ 2\sigma]}(x) dx.$$
\end{theorem}

Note that the above theorem only gives partial information about the behavior of the extremal eigenvalues (or about the operator norm) of $X_d$. For example, convergence in distribution implies that the larges eigenvalue of $X_d$ is at least 2 (which is the maximum of the support of the limit distribution $\mathrm{SC}_{0,1}$). The fact that the largest eigenvalue of $X_d$ converges indeed to 2 requires much more work, see \cite{bai1988necessary} for the case of Wigner matrices. In their seminal paper \cite{haagerup2005new}, Haagerup and Thorbj{\o}rnsen have further generalized these results to polynomials in tuples of GUE matrices and called this phenomenon \emph{strong convergence}.

\begin{definition}\label{def:strong-convergence}
A sequence of $k$-tuples of GUE distributed random matrices $(X^{(1)}_d, X_d^{(2)}, \ldots, X_d^{(k)}) \in \mathcal M_d^{sa}(\mathbb C)^k$ is said to \emph{converge strongly} towards a $k$-tuple of non-commutative random variables $(x_1, x_2, \ldots, x_k)$ living in some $C^*$- non-commutative probability space $(\mathcal A,\tau)$, if they converge in distribution: for all polynomials $P$ in $2k$ non-commutative variables,
$$\lim_{d \to \infty} \mathbb E \frac 1 d \operatorname{Tr} P(X^{(1)}_d, X^{(1)*}_d, \ldots,X^{(k)}_d, X^{(k)*}_d) = \tau [P(x_1, x_1^*, \ldots,x_k, x_k^*)]$$
and, moreover, for all $P$ as above, we also have the convergence of the operator norms:
$$\text{ almost surely, }\qquad \lim_{d \to \infty} \|P(X^{(1)}_d, X^{(1)*}_d, \ldots, X^{(k)}_d, X^{(k)*}_d)\| = \|P(x_1, x_1^*, \ldots, x_k, x_k^*)\|.$$
\end{definition}

Collins and Male generalized in \cite{collins2014strong} the result above to arbitrary unitarily invariant random matrices, by dropping the GUE hypothesis and asking that individual matrices $X_d^{(j)}$ converge strongly to their respective limits $x_j$, see \cite[Theorem 1.4]{collins2014strong}. Their result will be crucial to the present paper, since it will allow us to prove that the extremal eigenvalues have indeed the behavior suggested by the convergence in distribution (i.e.~they converge to the extrema of the support of the limiting eigenvalue distribution, in the single matrix case $k=1$).

\section{Some elements of free probability}\label{sec:free-prob}

We recall in this section the main tools from free probability theory needed here. The excellent monographs \cite{voiculescu1992free,nica2006lectures,mingo2017free} contain detailed presentations of the theory, with emphasis on different aspects. 

In free probability theory, non-commutative random variables are seen as abstract elements of some $C^*$-algebra $\mathcal A$, equipped with a trace $\tau$ which plays the role of the expectation in classical probability. The notion of distribution of a family of random variables $(x_1, \ldots, x_k)$ is the set of all evaluations $(P(x_1, x_1^*, \ldots, x_k,x_k^*))_P$, where $P$ runs through all polynomials in $2k$ non-commutative variables (see also Definition \ref{def:strong-convergence}). In the case of a single self-adjoint variable $x=x^*$, the distribution is given by the sequence of moments 
$$m_p(x) := \tau(x^p), \qquad p \geq 1.$$
The notion of \emph{free cumulants} introduced by Speicher in \cite{speicher1994multiplicative} plays a central role in the theory, in the sense that it characterizes free independence. In the case of a single variable, one can express the moments in terms of the free cumulants by the \emph{moment-cumulant formula}
$$m_p(x) = \sum_{\sigma \in \mathrm{NC}_p} \kappa_\sigma(x),$$
where the free cumulant functional $\kappa$ is defined multiplicatively on the cycles of the non-crossing partition $\sigma$:
$$\kappa_\sigma = \prod_{c \text{ cycle of }\sigma} \kappa_{|c|}.$$

Let us briefly discuss two examples. First, it is easy to see that the that the semicircular distribution introduced in Theorem \ref{thm:Wigner} has free cumulants $\kappa_1(\mathrm{SC_{a,\sigma}}) = a$, $\kappa_2(\mathrm{SC_{a,\sigma}}) = \sigma^2$, while $\kappa_p(\mathrm{SC_{a,\sigma}}) = 0$, for all $p \geq 3$. The vanishing of free cumulants of order 3 and larger characterizes the distribution which appears  in the free central limit theorem (exactly as in the classical situation, see \cite[Lecture 8]{nica2006lectures}).

Another remarkable family of distributions in free probability theory are the Mar{\v{c}}enko-Pastur distributions $\mathrm{MP}_c$, where $c>0$ is a positive scalar. The distribution $\mathrm{MP}_c$ is defined by the very simple property that all its free cumulants are equal to $c$: $\kappa_p(\mathrm{MP}_c) = c$, $\forall p \geq 1$. Using the moment-cumulant formula and Stieltjes inversion, one can compute the density of $\mathrm{MP}_c$:
\begin{equation}\label{eq:Marchenko-Pastur}
\mathrm{MP}_c=\max (1-c,0)\delta_0+\frac{\sqrt{(b-x)(x-a)}}{2\pi x} \; \mathbf{1}_{(a,b)}(x) \, \mathrm dx,
\end{equation}
where $a = (1-\sqrt c)^2$ and $b=(1+\sqrt c)^2$.
\begin{figure}[htbp]
\begin{center}
\includegraphics{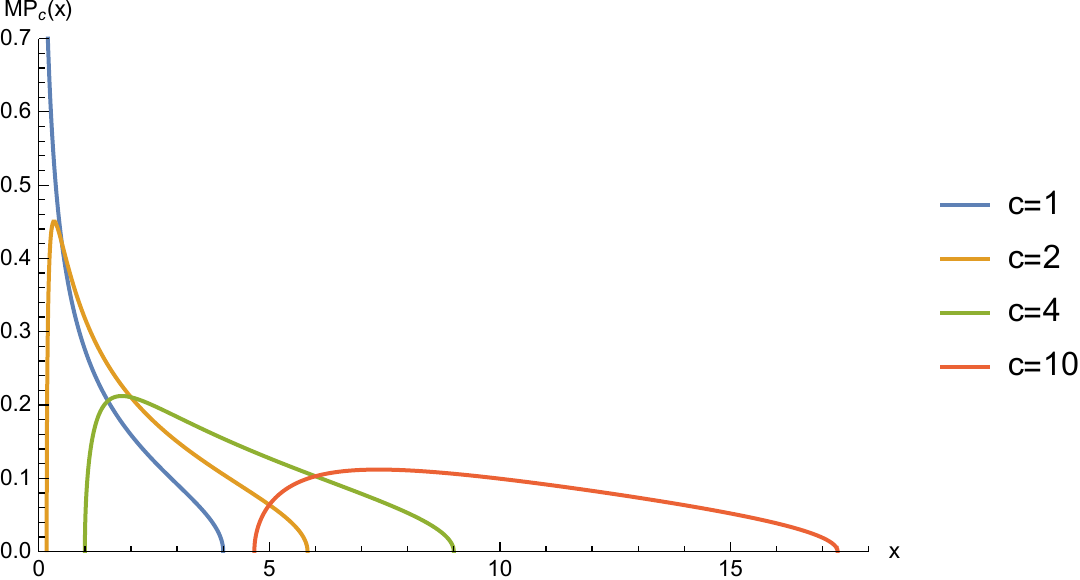}
\caption{The density of the Mar{\v{c}}enko-Pastur distributions $\mathrm{MP}_c$ for different values of the parameter $c$.}
\label{fig:Marchenko-Pastur}
\end{center}
\end{figure}

With the help of free cumulants, we can introduce the \emph{free additive convolution} of compactly supported probability measures, a notion which will play a key role in what follows. Given two compactly supported probability measures $\mu, \nu$, define $\mu \boxplus \nu$, the free additive convolution of $\mu$ and $\nu$, as the unique probability measure having free cumulants
$$\kappa_p(\mu \boxplus \nu) = \kappa_p(\mu) + \kappa_p(\nu) \qquad \forall p \geq 1.$$
For a given measure $\mu$, one can defined iteratively its free additive convolution powers as 
$$\mu^{\boxplus n}:= \underbrace{\mu \boxplus \cdots \boxplus \mu }_{n \text{ times}},$$
for any integer $n \geq 1$. As it was shown by Nica and Speicher in \cite{nica1996multiplication}, this semi-group extends from positive integers to all real numbers $T \geq 1$. This semi-group plays an important role in what follows, mainly due to the connection to block-modifications of random matrices (see Section \ref{sec:block-modified-strong}); for now, it is important to remember that the measures $\mu^{\boxplus T}$ are characterized by their free cumulants
$$\kappa_p(\mu^{\boxplus T}) = T \kappa_p(\mu) \qquad \forall p \geq 1, \, \forall T \in [1, \infty).$$

It is in general very hard to get a grip on the support of the elements of the free additive convolution semi-group $\mu^{\boxplus T}$. Although there exist implicit algebraic characterizations of the support of the measures $\mu^{\boxplus T}$ in terms of the support of $\mu$ and $T$, it is only in very simple circumstances that one can write down explicit formulas for the support. We recall below an approximation result obtained in \cite[Lemma 2.3 and Theorem 2.4]{collins2015estimates}. 

\begin{proposition}\label{prop:support-mu-T}
Let $\mu$ be a probability measure having mean $m$ and variance $\sigma^2$, whose support is contained in the compact interval $[A,B]$. Then, for any $T \geq 1$, we have
$$\operatorname{supp}(\mu^{\boxplus T}) \subseteq [A+m(T-1) - 2\sigma\sqrt{T-1}, B+m(T-1) + 2\sigma\sqrt{T-1}].$$
\end{proposition}

\section{The separability problem}\label{sec:entanglement}

We review in this section the notions of \emph{separability} and \emph{entanglement} from quantum information theory, as well as several important known results from this field. An excellent review of these notions is \cite{horodecki2009quantum}; for connections with random matrix theory, see \cite{collins2016random}.

We denote by $\mathcal M_d^+(\mathbb C)$ the cone of $d \times d$ complex positive semidefinite matrices. The \emph{separable cone} is a sub-cone of the set of bipartite  positive semidefinite matrices of size $d_1 \cdot d_2$ defined by
$$\mathrm{SEP}_{d_1,d_2} := \left \{\sum_{i=1}^k A_i \otimes B_i \, : \, A_i \in \mathcal M_{d_1}^+(\mathbb C), \, B_i \in \mathcal M_{d_2}^+(\mathbb C) \right \}.$$

Quantum states (resp.~separable quantum states) are elements of $\mathcal M_d^+(\mathbb C)$ (resp. $\mathrm{SEP}_{d_1,d_2}$) with unit trace; however, it is clear from the definition of separability that the trace normalization is of little importance, so we shall work with the conic versions of these notions to avoid technicalities.

Deciding whether a given positive semidefinite matrix $X \in \mathcal M_{d_1d_2}^+(\mathbb C)$ is separable is a NP-hard problem \cite{gurvits2003classical}, when formulated as a weak membership decision problem. A simple solution exists only in small dimensions $d_1d_2 \leq 6$, a fact due to the simple structure of the cone of positive maps $f:\mathcal M_{d_1}(\mathbb C) \to \mathcal M_{d_2}(\mathbb C)$. Indeed, any such positive map can be decomposed as (see \cite{woronowicz1976positive})
$$f = g + h \circ \top,$$
where $g,h$ are \emph{completely positive} maps and $\top$ is the transposition operator. Maps which can be written as above are called \emph{decomposable}; Woronowicz's result from \cite{woronowicz1976positive} shows that in the case $d_1d_2 \leq 6$, any positive map is decomposable; this fact is no longer true in larger dimensions, see \cite{horodecki1996separability}. 

The cone of separable matrices and the cone of positive maps 
$$\mathrm{POS}_{d_1, d_2} := \{f : \mathcal M_{d_1}(\mathbb C) \to \mathcal M_{d_2}(\mathbb C) \, : \,   A \geq 0 \implies f(A) \geq 0 \},$$ 
are dual to each other \cite{horodecki1996separability}:
$$X \in \mathrm{SEP}_{d_1, d_2} \iff \forall f \in \mathrm{POS}_{d_1, d_2} \quad  (f \otimes \mathrm{id}_{d_2})(X) \geq 0.$$

As we have already seen, the transposition map plays a special role in the theory. We introduce thus the cone $\mathrm{PPT}$ of matrices having a positive partial transpose
$$\mathrm{PPT}_{d_1,d_2} := \left \{X \in \mathcal M_{d_1d_2}^+(\mathbb C)\, : \, (\top_{d_1} \otimes \mathrm{id}_{d_2})(X) \geq 0 \right \}.$$
It is an intermediate cone, sitting between the separable cone and the positive semidefinite cone
$$\mathrm{SEP}_{d_1, d_2} \subseteq \mathrm{PPT}_{d_1, d_2} \subseteq \mathcal M_{d_1d_2}^+(\mathbb C).$$

\section{Strong convergence for block-modified random matrices}
\label{sec:block-modified-strong}

In this section we recall a result about the limiting distribution of random matrices obtained by acting with a given linear map on each block of a unitarily invariant random matrix \cite{arizmendi2016asymptotic}. We then upgrade this result to take into account strong convergence; the result will be used many times in the subsequent sections. 

The setting for block-modified random matrices is as follows. Consider a sequence of bipartite random matrices $X_d \in \mathcal M_{nd}^{sa}(\mathbb C)$ converging strongly as $d \to \infty$ to a compactly supported probability measure $\mu$ ($n$ being a fixed parameter). Given a (fixed) function $\varphi:\mathcal M_n(\mathbb C) \to \mathcal M_n(\mathbb C)$ preserving self-adjoint elements, define the \emph{modified random matrix} 
$$X_d^\varphi := (\varphi \otimes \mathrm{id}_d)(X_d) \in \mathcal M_{nd}^{sa}(\mathbb C),$$
obtained by acting with $\varphi$ on the $n \times n$ blocks of $X_d$. Note that in \cite{arizmendi2016asymptotic} the more general situation where $\varphi$ could change the size of blocks is considered, but this more general setting is not needed here. We also require that the function $\varphi$ satisfies the following technical condition (again, weaker conditions were considered in \cite{arizmendi2016asymptotic}; the situation here is closer to the results in \cite{banica2015block}), see \cite[Definition 4.7]{arizmendi2016asymptotic}.

\begin{definition}\label{def:UC}
Define the \emph{Choi matrix} of the linear map $\varphi$
$$\mathcal M_{n^2}^{sa}(\mathbb C) \ni C_\varphi := \sum_{i,j=1}^n \varphi(E_{ij}) \otimes E_{ij}.$$
The map $\varphi$ is said so satisfy the \emph{unitarity condition} if every eigenprojector $P$ of $C_\varphi$ satisfies
$$(\operatorname{id} \otimes \operatorname{Tr})(P) \sim I_n.$$
\end{definition}

Under this assumption on $\varphi$, we have the following result, which upgrades \cite[Theorem 5.1]{arizmendi2016asymptotic} to strong convergence. 

\begin{theorem}\label{thm:strong-cv-block-modified}
Consider a sequence of bipartite unitarily invariant random matrices $X_d \in \mathcal M_{nd}^{sa}(\mathbb C)$ converging strongly to a compactly supported probability measure $\mu$. Let $\varphi:\mathcal M_n(\mathbb C) \to \mathcal M_n(\mathbb C)$ be a hermiticity-preserving linear map satisfying the unitarity condition from Definition \ref{def:UC}. Then, the sequence of block-modified random matrices  $X_d^\varphi = (\varphi \otimes \mathrm{id}_d)(X_d)$ converges strongly to the probability measure 
\begin{equation}\label{eq:mu-varphi}
\mu^\varphi = \boxplus_{i=1}^s (D_{\lambda_i/n} \mu)^{\boxplus r_i},
\end{equation} 
where $\lambda_i$, resp~$r_i$, are the eigenvalues of the Choi matrix $C_\varphi$ and, respectively, their multiplicities. 
\end{theorem}
\begin{proof}
The convergence in distribution has been shown in  \cite[Theorem 5.1]{arizmendi2016asymptotic}. The strong convergence follows from \cite[Theorem 1.4]{collins2014strong} and the decomposition 
$$X_d^\varphi = \sum_{i,j,k,l=1}^n c_{ijkl} (E_{ij} \otimes I_d) X_d (E_{kl} \otimes I_d),$$
where $c_{ijkl} = \langle E_{il} \otimes E_{jk}, C_\varphi \rangle$. Indeed, the dilated matrix units $E_{ij} \otimes I_d$ are strongly asymptotically free from $X_d$, and the result follows. 
\end{proof}

\section{The partial transposition}
\label{sec:partial-transposition}

As an application of  Theorem \ref{thm:strong-cv-block-modified}, we consider in this section the operation of \emph{partial transposition}. Recall that transposition operation has the \emph{flip operator} as its Choi matrix: $F:\mathbb C^n \otimes \mathbb C^n \to \mathbb C^n \otimes \mathbb C^n$,
$$F x \otimes y = y \otimes x, \qquad \forall x,y \in \mathbb C^n.$$
The flip operator is unitary, having eigenvalues $+1, -1$ with respective multiplicities $n(n+1)/2$, $n(n-1)/2$ (the eigenvalues have as eigenspaces the symmetric, resp.~the antisymmetric subspace). 
\begin{proposition}
Let $X_d \in \mathcal M_{dn}^{+}(\mathbb C)$ a sequence of unitarily invariant random matrices as in Definition \ref{def:strong-convergence} converging strongly to a compactly supported probability measure $\mu \in \mathcal P([0, \infty))$; here, $n$ and $\mu$ are fixed. Define
\begin{equation}
\label{eq:def-mu-gamma} \mu^{\Gamma}:=(D_{1/n}\mu)^{\boxplus n(n+1)/2} \boxplus (D_{-1/n}\mu)^{\boxplus n(n-1)/2}
\end{equation}
If $\operatorname{minsupp} \mu^{\Gamma} >0$ then, almost surely as $d \to \infty$, $X_d \in \mathrm{PPT}_{n,d}$. In particular,
$$\lim_{d \to \infty} \mathbb P(X_d \in \mathrm{PPT}_{n,d}) = 1.$$
\end{proposition}
\begin{proof}
Using Theorem \ref{thm:strong-cv-block-modified}, the smallest eigenvalue of the partially transposed random matrix $X_d^\Gamma$ converges, almost surely as $d \to \infty$, towards $\operatorname{minsupp} \mu^{\Gamma}$, which is positive. Hence, the random matrices $X_d^\Gamma$ are asymptotically positive definite. 
\end{proof}

Let us discuss now some implications of this results. First, let consider some basic examples. Since GUE matrices are both unitarily invariant and Wigner, the result above applies to them, and we have the following remarkable equality ($\stackrel{\mathcal D}{=}$ denotes equality in distribution)
$$X_d \stackrel{\mathcal D}{=} X_d^\Gamma$$
for a GUE matrix $X_d \in \mathcal M_{nd}^{sa}(\mathbb C)$. In particular, we have that, for all $m \in \mathbb R$ and $\sigma \geq 0$, $\mathrm{SC}_{m,\sigma}^\Gamma = \mathrm{SC}_{m,\sigma}$. We show in the next lemma that semicircular measures are the only compactly supported probability measures enjoying this property. 

\begin{lemma}\label{lem:mu-equals-mu-gamma}
Assume $n \geq 2$ and let $\mu$ be a compactly supported probability measure such that $\mu^\Gamma = \mu$. Then, $\mu$ is semicircular. 
\end{lemma}
\begin{proof}
Let $\kappa_p$ be the free cumulants of the distribution $\mu$ (see Section \ref{sec:free-prob}) and 
$$R(z) =\sum_{p=0}^\infty \kappa_{p+1} z^p$$
be its $R$-transform. The equality of the two measures from the statement together with \eqref{eq:def-mu-gamma} give
$$R(z) = \frac{n+1}{2} R\left(\frac  z n\right) - \frac{n-1}2 R\left(-\frac z n \right).$$
On the level of the free cumulants, the equality above means that $\kappa_{p+1} = 0$ whenever 
$$\frac{n+1}{2n^p} - (-1)^p \frac{n-1}{2n^p} \neq 1.$$
Since $n \geq 2$, the above relation holds for all $p \geq 2$, so it must be that $\mu$ has only free cumulants of orders 1 and 2, and the conclusion follows. 
\end{proof}

Another interesting example for which one can perform computations is the case of the Mar{\v{c}}enko-Pastur distribution $\mathrm{MP}_c$, for some parameter $c>0$. This case has been studied in \cite[Theorem 6.2]{banica2013asymptotic}, where it has been shown that the measure $\mathrm{MP}_c^\Gamma$ has positive support iff
\begin{equation}\label{eq:threshold-PPT-MP}
c > 2  + 2 \sqrt{1-\frac1{n^2}}.
\end{equation}
As a remark, note that in the limit $n \to \infty$, we recover Aubrun's threshold value of $c=4$ from \cite[Theorems 2.2, 2.3]{aubrun2012partial}.

We prove next the main result of this section, a sufficient condition for the modified measure $\mu^\Gamma$ to be supported on the positive half-line.

\begin{theorem}\label{thm:bounds-mu-Gamma}
Let $\mu$ be a probability measure having mean $m$ and variance $\sigma^2$, whose support is contained in the compact interval $[A,B]$. Then, provided that $n(m-2\sigma) > B-A + 2\sigma$, we have $\operatorname{supp}(\mu^\Gamma) \subset (0,\infty)$.
\end{theorem}
\begin{proof}
We start by rewriting \eqref{eq:def-mu-gamma} as  
$$D_n \mu^\Gamma = \mu^{\boxplus n(n+1)/2} \boxplus D_{-1} \mu^{\boxplus n(n-1)/2} = \left( \mu^{\boxplus 1+\varepsilon} \boxplus D_{-1} \mu \right)^{\boxplus n(n-1)/2},$$
where $\varepsilon = 2/(n-1)$ is such that
$$(1+\varepsilon)\frac{n(n+1)}2 = \frac{n(n-1)}2.$$
Let us denote by $\nu$ the measure $\mu^{\boxplus 1+\varepsilon} \boxplus D_{-1} \mu$ and try to obtain bounds for its support. First, using Proposition \ref{prop:support-mu-T} for $T = 1+\varepsilon$, we get
$$\operatorname{supp}(\mu^{\boxplus 1 + \varepsilon}) \subseteq [A + m \varepsilon - 2\sigma \sqrt\varepsilon, B + m \varepsilon + 2\sigma \sqrt\varepsilon].$$
Thus, the support of $\nu$ is bounded from below by 
$$A_\nu := A + m \varepsilon - 2\sigma \sqrt\varepsilon - B.$$
Moreover, by direct computation, we have
\begin{align*}
\kappa_1(\nu) &= m \varepsilon\\
\kappa_2(\nu) &= \sigma^2(2+\varepsilon).
\end{align*}
Applying again Proposition \ref{prop:support-mu-T} for $\nu$ and $T = n(n-1)/2$. we deduce that the support of $D_n \mu^\Gamma$ is bounded from below by
\begin{align*}
A_\Gamma &=A -B +  m \varepsilon - 2\sigma \sqrt\varepsilon + m \varepsilon \left(\frac{n(n-1)}2 - 1 \right)  - 2\sigma \sqrt{2+\varepsilon}\sqrt{\frac{n(n-1)}2 - 1}\\
 &= nm - (B-A) - 2\sigma \left(\sqrt\frac{2}{n-1} + \sqrt\frac{(n-2)n(n+1)}{n-1} \right).
\end{align*}
The conclusion $A_\Gamma>0$ follows from the upper bound
$$\sqrt\frac{2}{n-1} + \sqrt\frac{(n-2)n(n+1)}{n-1}  <n+1,$$
which is satisfied for all $n\geq 2$.
\end{proof}

This result gives rather rough bounds for the semicircular and Mar{\v{c}}enko-Pastur distributions. For example, in the latter case, we obtain the condition $c>(2+6/n)^2$, which is off by a factor of roughly 2 from the exact bound \eqref{eq:threshold-PPT-MP}.

\section{Sufficient conditions - the depolarizing map} \label{sec:sufficient}

There are very few \emph{sufficient} conditions for the separability of a positive semidefinite matrix (or quantum state). For quantum states, the most famous one is the purity bound of Gurvits and Barnum \cite[Corollary 3]{gurvits2002largest}, corresponding to the fact that he in-radii of the convex sets of quantum states and separable states are identical. For the separable cone, this criterion reads: given a positive semidefinite matrix $X \in \mathcal M_{dn}(\mathbb C)$, $X \neq 0$
$$\frac{\operatorname{Tr}(X^2)}{(\operatorname{Tr} X)^2} \leq \frac{1}{nd-1} \implies X \in \mathrm{SEP}_{n,d}.$$
Note however that the condition above is very restrictive: by the Cauchy-Schwarz inequality, we always have 
$$\frac{1}{nd} \leq \frac{\operatorname{Tr}(X^2)}{(\operatorname{Tr} X)^2}.$$
In particular, if we consider a sequence of random matrices converging strongly (as in Definition \ref{def:strong-convergence} to a probability measure $\mu$, the only case in which the Gurvits-Barnum condition would hold is when $\mu$ had 0 variance, that is $X$ would be closer to a multiple of the identity matrix. 

We consider next a more powerful separability criterion, given by the \emph{depolarizing channel}. Recall that the depolarizing channel of parameter $t \in [-1/(n^2-1), 1]$ is the completely positive, trace preserving map $\Delta_t : \mathcal M_n(\mathbb C) \to \mathcal M_n(\mathbb C)$ given by
$$\Delta_t(X) = tX + (1-t)\frac{I_n}{n}.$$
It is known that the quantum channel $\Delta_t$ is entanglement breaking iff $t \in [-1/(n^2-1), 1/(n+1)]$ \cite[Section V]{horodecki1999reduction}. This means that, when the parameter $t$ lies inside the above specified range, we have, for all positive semidefinite input matrices $Y \in \mathcal M_{nd}^+(\mathbb C)$, 
$$(\Delta_t \otimes \operatorname{id})(Y) \in \mathrm{SEP}_{n,d}.$$
Using this observation, we obtain the following sufficient separability conditions. 

\begin{proposition}\label{prop:separability-depolarizing}
Let $X \in \mathcal M_{nd}^+(\mathbb C)$ be a positive semidefinite operator. If any of the two conditions below is satisfied, then $X \in \mathrm{SEP}_{n,d}$:
\begin{align}
\label{eq:cond-suff-sep-+}(n+1) X &\geq I_n \otimes (\operatorname{Tr}_n \otimes \operatorname{id}_d)(X)\\
\label{eq:cond-suff-sep--}(n^2-1)X &\leq n  I_n \otimes (\operatorname{Tr}_n \otimes \operatorname{id}_d)(X).
\end{align}
\end{proposition}
\begin{proof}
For a given $X$, let us solve the equation $(\Delta_t \otimes \operatorname{id})(Y) = X$. Writing $Y_2:=(\operatorname{Tr}_n \otimes \operatorname{id}_d)(Y)$ for the partial trace of $Y$ with respect to the first tensor factor, we have
$$tY  +\frac{1-t}{n} I_n \otimes Y_2 = X.$$
 Taking the partial trace of this equation with respect to the first factor, we get $X_2 = Y_2$. Plugging this in, we finally obtain
 $$tY = X - \frac{I_n}{n} \otimes X_2.$$
If $t=0$, the condition above implies that $X$ is of the form $X = \frac{I_n}{n} \otimes X_2$. For $t>0$, asking that $Y \geq 0$ amounts to having
$$X \geq (1-t)\frac{I_n}{n} \otimes X_2.$$
The weakest necessary condition is obtained when $t$ takes the largest possible value (for which $\Delta_t$ is still entanglement breaking), that is $t = 1/(n+1)$. The condition reads then
$$(n+1)X \geq I_n \otimes (\operatorname{Tr}_n \otimes \operatorname{id}_d)(X),$$
which is the first condition announced. To recapitulate, for $X$ satisfying the condition above, there exist a positive semidefinite matrix $Y$ such that $X=(\Delta_{1/(n+1)} \otimes \operatorname{id})(Y)$. Since the quantum channel $\Delta_{1/(n+1)}$ is entanglement breaking \cite[Section V]{horodecki1999reduction}, the output matrix $X$ is separable. Similarly, for negative values of $t$, we obtain the condition \eqref{eq:cond-suff-sep--}, finishing the proof.
\end{proof}
Before we move on, let us present a second point of view on the condition \eqref{eq:cond-suff-sep-+}. Note that
$$\frac{2}{n+1} \int_{\|x\|=1} \langle x, X x\rangle \, |x \rangle \langle x | dx = \frac{1}{n+1}(X + I_n) = \Delta_{1/(n+1)}(X).$$
Obviously, the left hand side of the equality above defines an entanglement breaking channel; one can generalize this idea, by considering the more general \emph{measure and prepare map}
$$\operatorname{MP}_p(X) = \int_{\|x\|=1} \langle x, X x\rangle^p \, |x \rangle \langle x | dx,$$
for some positive integer $p \geq 1$. It is clear that the (non-linear) map $(\mathrm{id} \otimes \mathrm{MP}_p)$ has a separable range (when restricted to the PSD cone). It is however more challenging to invert this map; as an example, we have, for $p=2$
$$\operatorname{MP}_2(X) = [\operatorname{Tr}(X^2)I_n + \operatorname{Tr}(X) X + X^2]/3.$$
Such maps appear in the context of quantum de Finetti theorems \cite{renner2007symmetry,christandl2007one} and the exchangeability separability hierarchy \cite{doherty2004completea}.

\begin{theorem}\label{thm:sufficient-conditions-mu}
Let $X_d \in \mathcal M_{dn}^{+}(\mathbb C)$ a sequence of unitarily invariant random matrices as in Definition \ref{def:strong-convergence} converging strongly to a compactly supported probability measure $\mu \in \mathcal P([0, \infty))$; here, $n$ and $\mu$ are fixed. Define
\begin{align}
\label{eq:def-mu-delta+} \mu^{\Delta+}&:=D_{\frac{n^2+n-1}{n}}\mu \boxplus (D_{-1/n}\mu)^{\boxplus(n^2-1)}\\
\label{eq:def-mu-delta-} \mu^{\Delta-}&:=D_{2-n^2}\mu \boxplus \mu^{\boxplus(n^2-1)}.
\end{align}
If $\operatorname{minsupp} \mu^{\Delta+} >0$ or $\operatorname{minsupp} \mu^{\Delta-} >0$ then, almost surely as $d \to \infty$, $X_d \in \mathrm{SEP}_{n,d}$; in particular,
$$\lim_{d \to \infty} \mathbb P(X_d \in \mathrm{SEP}_{n,d}) = 1.$$
\end{theorem}
\begin{proof}
The proof uses the conditions \eqref{eq:cond-suff-sep-+} and \eqref{eq:cond-suff-sep--} and Theorem \ref{thm:strong-cv-block-modified}. Let us work through the first case, the second one being similar. The sufficient condition \eqref{eq:cond-suff-sep-+} for separability is equivalent to $(\varphi_+ \otimes \operatorname{id}_d)(X_d) \geq 0$, for the map $\varphi_+:\mathcal M_n(\mathbb C) \to \mathcal M_n(\mathbb C)$ given by
$$\varphi_+(X) = (n+1)X - (\operatorname{Tr} X)I_n.$$
This map satisfies the unitarity condition from Definition \ref{def:UC}, where the Choi matrix of $\varphi_+$ has eigenvalues
\begin{align*} 
n(n+1) - 1,\qquad &\text{ with multiplicity } 1\\
-1,\qquad &\text{ with multiplicity } n^2-1.
\end{align*}
Hence, by Theorem \ref{thm:strong-cv-block-modified}, the random matrices $(\varphi_+ \otimes \operatorname{id}_d)(X_d)$ converge strongly, as $d \to \infty$, towards the probability measure $\mu^{\Delta+}$ from \eqref{eq:def-mu-delta+}. The positivity of the support of $\mu^{\Delta+}$ ensures that the random matrices $(\varphi_+ \otimes \operatorname{id}_d)(X_d)$ are asymptotically positive definite. 
\end{proof}

Let us consider some examples. For the semicircular distribution with mean $m$ and variance $\sigma^2$, we get by direct computation $\mathrm{SC}_{m, \sigma}^{\Delta\pm} = \mathrm{SC}_{m, \sigma_\pm}$ with
$$\sigma_+ = \sigma\frac{\sqrt{n^4+2n^3-2n}}{n} \qquad \text{ and  } \qquad \sigma_- = \sigma\sqrt{2n^2-3}.$$
In this case, since both modified measures are semicircular and have the same average, the criterion is stronger when the standard deviation is smaller. In the range $n \geq 2$, we have $\sigma_+ \leq \sigma_-$ iff $n \geq 3$. Indeed, the inequality simplifies to $n^3 -2n^2-3n+2 \geq 0$. The above polynomial changes signs 3 times between the values $-2, 0,2,3$, thus proving the claim (the actual roots of this polynomial are approximately $-1.34292$, $0.529317$, $2.81361$). In the case $n=2$, $\sigma_+ > \sigma_- = \sigma \sqrt 5$. For shifted GUE random matrices, we have the following result. 

\begin{corollary}
Let $Y_d = 2I_{nd} + \alpha X_d \in M_{nd}^{sa}(\mathbb C)$ be a sequence of random matrices, where $X_d$ is a standard GUE and $\alpha \in (0,1)$ is a fixed parameter. Then, provided that 
$$\alpha < \begin{cases}
\frac{1}{\sqrt 5},&\quad \text{ if } n=2\\
\frac{n}{\sqrt{n^4+2n^3-2n}},&\quad \text{ if } n \geq 3,
\end{cases}$$
the random matrices $Y_d$ are almost surely asymptotically $n \otimes d$ separable. In particular. 
$$\lim_{d \to \infty} \mathbb P(Y_d \in \mathrm{SEP}_{n,d}) = 1.$$
\end{corollary}

Let us now apply Theorem \ref{thm:sufficient-conditions-mu} to the case of the Mar{\v{c}}enko-Pastur distribution. Although we are not able to determine analitically the support of the probability distributions $\mathrm{MP}_c^{\Delta \pm}$, we present some useful bonds. 

\begin{corollary}
Let $X_d  \in M_{nd}^{sa}(\mathbb C)$ be a sequence of unitarily invariant random matrices as in Definition \ref{def:strong-convergence} converging strongly to the Mar{\v{c}}enko-Pastur probability distribution of parameter $c \geq 1$. Then, provided that 
\begin{equation}\label{eq:sufficient-SEP-MP}
\sqrt c > \frac{1 + \sqrt{n^2+n-1}}{\sqrt{n^2+n-1} - \sqrt{n^2-1}} \sim 2n,
\end{equation}
the random matrices $X_d$ are almost surely asymptotically $n \otimes d$ separable. In particular. 
$$\lim_{d \to \infty} \mathbb P(X_d \in \mathrm{SEP}_{n,d}) = 1.$$
\end{corollary}
\begin{proof}
Let us consider first the criterion corresponding to the map $\Delta_+$. We are interested in the support of the probability measure 
$$\mathrm{MP}_c^{\Delta+} := D_{\frac{n^2+n-1}{n}}\mathrm{MP}_c \boxplus (D_{-1/n}\mathrm{MP}_c^{\boxplus n^2-1}) = D_{\frac{n^2+n-1}{n}}\mathrm{MP}_c \boxplus D_{-1/n} \mathrm{MP}_{c(n^2-1)}.$$
A sufficient condition for the support of $\mathrm{MP}_c^{\Delta+}$ to be positive is that 
$$\frac{n^2+n-1}{n} \operatorname{minsupp} \mathrm{MP}_c > \frac 1 n \operatorname{maxsupp} \mathrm{MP}_{c(n^2-1)},$$
which is equivalent to \eqref{eq:sufficient-SEP-MP}. A similar analysis for the map $\Delta_-$ yields the sufficient condition 
$$\sqrt c > \frac{1 + \sqrt{n^2-2}}{\sqrt{n^2-1} - \sqrt{n^2-2}},$$
which can be seen to be weaker than \eqref{eq:sufficient-SEP-MP} for $n \geq 2$, proving the claim.
\end{proof}

Following the proof of Lemma \ref{lem:mu-equals-mu-gamma}, one can easily show that the only probability distributions which are invariant under the $\Delta_\pm$ modifications are Dirac masses; we leave the proof as an exercise for the reader. 

Let us now look for sufficient conditions on the probability measure $\mu$ which would ensure that the hypotheses of Theorem \ref{thm:sufficient-conditions-mu} are satisfied. Our approach here is identical to the one used in Theorem \ref{thm:bounds-mu-Gamma}.

\begin{proposition}\label{prop:bounds-mu-Delta}
Let $\mu$ be a probability measure having mean $m$ and variance $\sigma^2$, whose support is contained in the compact interval $[A,B]$. Then, provided that 
$$(n^2+n-1)A > B + m(n^2-2) + 2\sigma \sqrt{n^2-2},$$
we have $\operatorname{supp}(\mu^{\Delta+}) \subset (0, \infty).$ Similarly,
$$A>(n^2-2)(B-m) + 2\sigma\sqrt{n^2-2} \implies \operatorname{supp}(\mu^{\Delta-}) \subset (0, \infty).$$
In particular, if any of the conditions above hold, then, almost surely as $d \to \infty$, $X_d \in \mathrm{SEP}_{n,d}$; in particular,
$$\lim_{d \to \infty} \mathbb P(X_d \in \mathrm{SEP}_{n,d}) = 1.$$
\end{proposition}
\begin{proof}
We just prove the implication for $\mu^{\Delta+}$, the other one being similar. Let us define 
$$A_1 := \operatorname{minsupp}\left(D_{\frac{n^2+n-1}{n}}\mu\right) \quad \text{ and } \quad A_2 := \operatorname{minsupp}\left(D_{-1/n}\mu^{\boxplus(n^2-1)}\right).$$
We obviously have $A_1 \geq A (n^2+n-1)/n$; to lower bound $A_2$, we use Proposition \ref{prop:support-mu-T} for $D_{-1/n}\mu$, to obtain
$$A_2 \geq \frac{-B}{n} + \frac{-m}{n}(n^2-2) - 2 \frac{\sigma}{n}\sqrt{n^2-2}.$$
The conclusion follows now from the previous two inequalities, ensuring that $A_1+A_2 > 0$.
\end{proof}

\section{Necessary conditions - the correlated witness} \label{sec:necessary}

We shift focus in this section and study \emph{necessary} conditions for separability, or, equivalently, sufficient conditions for entanglement. Many such criteria (usually called entanglement criteria), exist in the literature, and we shall start by quickly reviewing them. Next, we discuss a criterion coming from a random entanglement witness, arguing that is a very useful one. 

Given the use of entanglement for quantum tasks, and the computational hardness of deciding separability, there exist a plethora of criteria permitting to certify the entanglement of a given (mixed) quantum state. Most of these criteria stem from the following very simple observation: Let $f:\mathcal M_n(\mathbb C) \to \mathcal M_n(\mathbb C)$ be a positive map (that is, a map which preserves the positive semidefinite cone). Then, for any matrix $X \in \mathrm{SEP}_{n,d}$, we have 
$$(f \otimes \operatorname{id})(X) \geq 0.$$
Hence, if the output $(f \otimes \operatorname{id})(Y)$ is not positive semidefinite, then the input matrix $Y$ is entangled (assuming that $Y$ was positive semidefinite to begin with). Every choice of a positive map $f$ yields an entanglement criterion; some of the most studied such maps are the \emph{transposition map} (giving the PPT criterion discussed at the end of Section \ref{sec:block-modified-strong}) and the reduction map 
$$f(X) = (\operatorname{Tr} X)I_n - X,$$
giving the \emph{reduction criterion} \cite{cerf1999reduction,horodecki1999reduction}, which can be shown to be weaker (i.e.~it detects fewer entangled states) than the PPT criterion, but it is interesting nonetheless for its relation to the distillability problem. There are some other entanglement criteria which do not fall in this framework, the most notable being the realignment criterion \cite{chen2002matrix,rudolph2003cross}; we shall not discuss these criteria here, see \cite{horodecki2009quantum} for a review and \cite{aubrun2012realigning,jivulescu2015thresholds} for results about random quantum states. 

Since the set of separable states is a closed convex cone, by the Hahn-Banach theorem one can find, given any entangled matrix $X$, one can find a hyperplane separating $X$ from $\mathrm{SEP}_{n,d}$. In other words, there exists a \emph{block-positive} operator $W \in \mathcal M_{nd}^{sa}(C)$, called an \emph{entanglement witness}, such that $\operatorname{Tr}(WX) <0$. We recall that an operator $W$ is called block-positive iff 
\begin{equation}\label{eq:def-BP}
\forall x \in \mathbb C^n, \forall y \in \mathbb C^d, \quad \langle x \otimes y, W x \otimes y \rangle \geq 0.
\end{equation}

In this section, we shall make a very particular choice for the operator $W$: we shall set, for a constant $\beta \in \mathbb R$, 
\begin{equation}\label{eq:def-WX}
W:=\beta I_{nd} - X,
\end{equation}
using the following intuition: \emph{what better witness for a quantum state's entanglement than the state itself?} Pursuing this idea for unitarily invariant quantum states, we obtain the entanglement criterion from Theorem \ref{thm:witness-X}. Before we state and prove that theorem, we need some preliminary results, which we find interesting for their own sake. 

First, let us recall the following definition from \cite{johnston2010family}, see also \cite{johnston2011family}:

\begin{definition}
The \emph{$S(k)$ norm} of an operator $X \in \mathcal M_{nd}(\mathbb C)$ is defined to be
\begin{equation}\label{eq:def-Sk}
\|X\|_{S(k)} := \sup \{ |\langle v, X w\rangle| \, : \, \operatorname{SR}(v), \operatorname{SR}(w) \leq k\},
\end{equation}
where the \emph{Schmidt rank} of a vector $v \in \mathbb C^n \otimes \mathbb C^d$ is its tensor rank
$$\operatorname{SR}(v) := \min \{ k \geq 0 \, : \, v = \sum_{i=1}^k x_i \otimes y_i \}.$$
If the operator $X$ is normal, than one can restrict the maximization in \eqref{eq:def-Sk} to $w=v$.
\end{definition}

Obviously, the operator $W$ from \eqref{eq:def-WX} is block-positive as soon as $\beta \geq \|X\|_{S(1)}$ (moreover, if $X$ were positive, then the two statements would be equivalent, see \cite[Corollary 4.9]{johnston2010family} for the general case of $k$-block-positivity). So, in order to certify the block-positivity of bipartite operators having a strong asymptotic limit, we need the following result. 

\begin{proposition}\label{prop:Sk-norm-ui}
Let $X_d \in \mathcal M_{dn}^{+}(\mathbb C)$ a sequence of unitarily invariant random matrices as in Definition \ref{def:strong-convergence} converging strongly to a compactly supported probability measure $\mu \in \mathcal P(\mathbb R)$; here, $n$ and $\mu$ are fixed.  Then, almost surely, 
$$\lim_{d \to \infty} \|X_d\|_{S(k)} = \frac k n \|\mu^{\boxplus n/k}\|_\infty,$$
where we write $\|\nu\|_\infty := \|A\|_{L^\infty}$ for some random variable $A$ having distribution $\nu$.
\end{proposition}
\begin{proof}

Since we are interested in the limit $d \to \infty$ and $n$ is fixed, we assume wlog that $n \leq d$. Moreover, since the matrices $X_d$ are self-adjoint, we have 
$$\|X_d\|_{S(k)} = \max(|m_d|,|M_d|),$$
where
\begin{align*}
m_d&:= \inf_{v\in \mathbb C^{nd}, \|v\| = 1, \operatorname{SR}(v) \leq k} \langle v, X_d v \rangle\\
M_d&:= \sup_{v\in \mathbb C^{nd}, \|v\| = 1, \operatorname{SR}(v) \leq k} \langle v, X_d v \rangle.
\end{align*}
We relate now the above numbers to $k$-positivity: 
\begin{align*}
m_d&= \sup\{ \lambda \in \mathbb R \, : \, X_d - \lambda \text{ is $k$-positive}\}\\
M_d&=  \inf\{ \lambda \in \mathbb R \, : \, \lambda - X_d \text{ is $k$-positive}\}.
\end{align*}
The asymptotic $k$-positivity of strongly convergent sequences of random matrices has been studied in \cite[Theorem 4.2]{collins2015random}, where it has been shown that, almost surely, the sequence $X_d$ is asymptotically $k$-block-positive if $\operatorname{supp}(\mu^{\boxplus n/k}) \subset (0, \infty)$, and, reciprocally, it is not $k$-block-positive if $\operatorname{supp}(\mu^{\boxplus n/k}) \cap (-\infty, 0) \neq \emptyset$; note that the case where the left endpoint of the support of $\mu^{\boxplus n/k}$ is zero is excluded, since in this case one needs extra information about the fluctuations of the smallest eigenvalue. Applying this result to our setting, we obtain, say for $m_d$: almost surely,
$$m:=\lim_{d \to \infty} m_d = \sup\{ \lambda \in \mathbb R \, : \, \operatorname{supp}\left(\mu_\lambda^{\boxplus n/k}\right) \subset (0,\infty)\},$$
where $\mu_\lambda = T_{-\lambda}\mu$, with $T_\cdot$ denoting the translation operator. We have obviously
$$(T_{-\lambda} \mu)^{\boxplus n/k} = T_{-\lambda n/k}(\mu^{\boxplus n/k}),$$
and thus
\begin{align*}
m&= \sup\{ \lambda \in \mathbb R \, : \, \operatorname{supp}\left(\mu_\lambda^{\boxplus n/k}\right) \subset (0,\infty)\}\\
&=  \sup\{ \lambda \in \mathbb R \, : \, \operatorname{supp}\left(T_{-\lambda n/k} \left(\mu^{\boxplus n/k}\right)\right) \subset (0,\infty)\}\\
&=  \sup\{ \lambda \in \mathbb R \, : \, \operatorname{supp}\left(\mu^{\boxplus n/k}\right) \subset (\lambda n/k,\infty)\}\\
&= \frac k n \operatorname{minsupp}(\mu^{\boxplus n/k}).
\end{align*}
Similarly, we get
$$M:=\lim_{d \to \infty} M_d = \frac k n \operatorname{maxsupp}(\mu^{\boxplus n/k}),$$
finishing the proof.
\end{proof}

Before moving on, let us discuss the value of the $S(k)$ norms for random projections. This case is important for quantum information theory, as it was argued in \cite[Section 7] {johnston2011family}; see also \cite[Theorem 4.15]{johnston2010family} for general norm bonds for projections. We consider here a sequence of Haar-distributed random projection operators $P_d \in \mathcal M_{nd}^{sa}(\mathbb C)$ of ranks $r_d \sim \rho nd$ for some fixed parameter $\rho \in (0,1)$. Using \cite[Proposition 2.9]{fukuda2015additivity}, we obtain the following asymptotic behavior. 
\begin{corollary} 
For a sequence $(P_d)_d$ of random projections as above, and for any $1 \leq k \leq n$, we have the following almost sure limit:
$$\lim_{d \to \infty} \|P_d \|_{S(k)} = \begin{cases} 
1,&\qquad \text{ if } \rho > 1 - \frac k n\\
\rho + \frac k n - 2 \rho \frac k n + 2 \sqrt{\frac k n \left(1 - \frac k n\right) \rho (1-\rho)},&\qquad \text{ if } \rho \leq 1 - \frac k n.\end{cases}$$
\end{corollary}

We have now all the ingredients to state and prove the main result of this section.

\begin{theorem}\label{thm:witness-X}
Let $X_d \in \mathcal M_{dn}^{+}(\mathbb C)$ a sequence of unitarily invariant random matrices as in Definition \ref{def:strong-convergence} converging to a compactly supported probability measure $\mu \in \mathcal P([0, \infty))$; here, $n$ and $\mu$ are fixed. If 
\begin{equation}\label{eq:witness-X}
\frac 1 n \operatorname{maxsupp}\left( \mu^{\boxplus n}\right) < \frac{m_2(\mu)}{m_1(\mu)}
\end{equation}
 then, almost surely as $d \to \infty$, $X_d \notin \mathrm{SEP}_{n,d}$. In particular,
$$\lim_{d \to \infty} \mathbb P(X_d \in \mathrm{SEP}_{n,d}) = 0.$$
\end{theorem}
\begin{proof}
To show that the matrices $W_d$ from \eqref{eq:def-WX} are indeed entanglement witnesses for $X_d$ (almost surely as $d \to \infty$), we need to show, for an appropriate choice of the constant $\beta$, two things:
\begin{enumerate}
\item The maps $W_d$ are asymptotically block-positive.
\item $\lim_{d \to \infty} \langle W_d, X_d \rangle < 0$.
\end{enumerate}
We use Proposition \ref{prop:Sk-norm-ui} with $k=1$ for the first item: $W_d$ are asymptotically entanglement witnesses provided that 
$$\beta > \lim_{d \to \infty} \|X_d\|_{S(1)} = \|\mu^{\boxplus n}\|_\infty = \frac 1 n \operatorname{maxsupp}(\mu^{\boxplus n}).$$ 
The computation of the limit appearing in the second item above is straightforward: almost surely, we have
$$\lim_{d \to \infty} \frac{1}{nd} \langle \beta I_{nd} - X_d, X_d \rangle  = \beta m_1(\mu) - m_2(\mu).$$
We are done: choose any $\beta$ satisfying
$$\frac 1 n \operatorname{maxsupp}(\mu^{\boxplus n}) < \beta < \frac{m_2(\mu)}{m_1(\mu)}.$$
\end{proof}

As in the previous section, we consider next some applications of the result above, which we state as corollaries. We start with the case of shifted GUEs, see also \cite[Theorem 5.4]{collins2015random}.

\begin{corollary}
Let $X_d \in \mathcal M_{nd}^{sa}(\mathbb C)$ a sequence of (normalized) GUE matrices, and set $Y_d := mI_{nd} + \sigma X_d$, for some constants $m, \sigma \geq 0$.  If 
$$\frac 1 2 < \frac \sigma m < \frac 2 {\sqrt n},$$
then $Y_d$ is asymptotically positive semidefinite, PPT, and entangled. 
\end{corollary}
\begin{proof}
Note that the sequence $Y_d$ from the statement converges strongly to the semicircular probability measure $\operatorname{SC}_{m,\sigma}$, which is supported on the interval $[m-2\sigma, m+2\sigma]$. Hence, if $\sigma / m > 1/2$, the matrices $Y_d$ are asymptotically positive semidefinite, and also PPT (since the GUE distribution is Wigner). For the second inequality, use $\operatorname{SC}_{m,\sigma}^{\boxplus n} = \operatorname{SC}_{mn,\sigma\sqrt n}$.
\end{proof}

\begin{corollary}
Let $X_d \in \mathcal M_{dn}^{+}(\mathbb C)$ a sequence of unitarily invariant random matrices as in Definition \ref{def:strong-convergence} converging strongly to the Mar{\v{c}}enko-Pastur probability distribution of parameter $c>0$. If 
$$c < \frac{(n-1)^2}{4n}$$
then, almost surely as $d \to \infty$, $X_d \notin \mathrm{SEP}_{n,d}$; in particular,
$$\lim_{d \to \infty} \mathbb P(X_d \in \mathrm{SEP}_{n,d}) = 0.$$
\end{corollary}
\begin{proof}
We use the criterion in Theorem \ref{thm:witness-X} to obtain the following condition for entanglement
$$\frac{(\sqrt{cn}+1)^2}{n} = \frac 1 n \operatorname{maxsupp}\left( \mathrm{MP}_c^{\boxplus n}\right) < \frac{m_2(\mathrm{MP}_c)}{m_1(\mathrm{MP}_c)} = \frac{c^2+c}{c}.$$
\end{proof}

Putting together the bounds above with the ones from \cite[Theorem 6.2]{banica2013asymptotic} (see also \eqref{eq:threshold-PPT-MP}, we obtain the following corollary. 

\begin{corollary}
For any $n \geq 18$ and $c$ such that 
$$c \in \left( 2 + 2\sqrt{1 - \frac{1}{n^2}}, \frac{(n-1)^2}{4n} \right),$$
a sequence of unitarily invariant random matrices  $X_d \in \mathcal M_{dn}^{+}(\mathbb C)$ converging strongly to the Mar{\v{c}}enko-Pastur probability distribution of parameter $c>0$ is, almost surely in the limit $d \to \infty$, PPT and entangled. 
\end{corollary}

\begin{proposition}\label{prop:bounds-X}
Let $\mu$ be a probability measure having mean $m$ and variance $\sigma^2$, whose support is contained in the compact interval $[A,B]$. Assume that 
$$\frac B m < 1 + n \frac{\sigma^2}{m^2} - 2 \frac{\sigma}{m}\sqrt{n-1}.$$
Then, for any sequence of unitarily invariant random matrices  $X_d \in \mathcal M_{dn}^{+}(\mathbb C)$ converging strongly to $\mu$, we have that almost surely as $d \to \infty$, $X_d \notin \mathrm{SEP}_{n,d}$; in particular,
$$\lim_{d \to \infty} \mathbb P(X_d \in \mathrm{SEP}_{n,d}) = 0.$$
\end{proposition}
\begin{proof}
The result follows from Theorem \ref{thm:witness-X}, using the upper bound from Proposition \ref{prop:support-mu-T}.
\end{proof}

\section{PPT matrices with large Schmidt number}\label{sec:SN-PPT}

The Schmidt number of a positive semidefinite matrix $X \in \mathcal M_{d_1}(\mathbb C) \otimes \mathcal M_{d_2}(\mathbb C)$ is a discrete measure of entanglement. It is defined, for rank-one matrices as 
$$\operatorname{SN}(xx^*) = \operatorname{rk} [\operatorname{id}_{d_1} \otimes \operatorname{Tr}_{d_2}](xx^*)$$
and extended by the convex roof construction to arbitrary matrices
$$\operatorname{SN}(X) = \min\{ r \, : \, X = \sum_{i=1}^m x_i x_i^* \, \text{ with } \, \operatorname{SN}(x_ix_i^*) \leq r\}.$$ 
Obviously, $\operatorname{SN}(X) =1$ iff $X \in \mathrm{SEP}_{d_1,d_2}$, and $\operatorname{SN}(X) \leq \min(d_1,d_2)$ for all positive semidefinite $X$. It is an interesting question whether imposing that the partial transposition of $X$ is positive semidefinite has any implications on the range of values the Schmidt number can take. Very recently, explicit examples of PPT matrices $X \in \mathcal M_d(\mathbb C) \otimes \mathcal M_d(\mathbb C)$ with $\operatorname{SN}(X) \geq 	\lceil (d-1)/4 \rceil$ have been constructed \cite[Corollary III.3]{huber2018high}. In the same paper, the authors show that, for large $d$, most quantum states acting on $\mathbb C^d \otimes \mathbb C^d$ have Schmidt number greater than $cd$, for some universal constant $c$. 

In the unbalanced case, we show that the linear scaling $\operatorname{SN}(X) \geq \min(d_1,d_2)/16$ can be achieved by using GUE random matrices. The example below complements the construction of PPT entangled states from \cite[Section 5]{collins2015random}, by providing a lower bound for the Schmidt number.

\begin{theorem}
For any fixed integer $n \geq 2$, consider the sequence of self-adjoint matrices $X_d:=aI_{nd} -  G_d \in \mathcal M_n(\mathbb C) \otimes \mathcal M_d(\mathbb C)$, where $G_d$ is a $\mathrm{GUE}_{nd}$ random matrix. There exists a constant $a >0$ (made explicit in the proof) such that the following conditions hold almost surely, as $d \to \infty$:
\begin{itemize}
\item $X_d$ is PPT: $X_d, X_d^\Gamma \geq 0$
\item $\operatorname{SN}(X_d) > \lfloor (n-1)/16 \rfloor$.
\end{itemize}
\end{theorem}
\begin{proof}
The asymptotic distribution of the random matrix $X_d$ is $\mathrm{SC}_{a,1}$, and thus $X_d$ is positive semidefinite as $d \to \infty$ iff 
\begin{equation}\label{eq:SN-PPT-a-sigma}
a > 2.
\end{equation}
Recall from Section \ref{sec:partial-transposition} that the matrices $X_d$ and $X_d^\Gamma$ have the same distribution, so the fact that $X_d^\Gamma$ is also positive semidefinite comes at no cost (this being the reason that shifted GUE random matrices are useful for PPT-related questions). 

Let us now show that, asymptotically, $\operatorname{SN}(X_d) > \lfloor (n-1)/16 \rfloor$. This relation is equivalent to finding a $\lfloor (n-1)/16 \rfloor$-positive map $\Phi_d:\mathcal M_n(\mathbb C) \to \mathcal M_d(\mathbb C)$ such that $[\Phi_d \otimes \operatorname{id}_d](X_d)$ is not positive semidefinite. Let $C_d \in \mathcal M_n(\mathbb C) \otimes \mathcal M_d(\mathbb C)$ denote the Choi matrix of the adjoint map $\Phi_d^*$, and let us choose $C_d = bI_{nd} + G_d$. Importantly, the matrix $G_d$ here is the same as the one appearing in the definition of the matrix $X_d$; hence, the random matrix $X_d$ and the random map $\Phi_d$ are correlated. Note that the distribution of the Choi matrix $C_d$ is $\mathrm{SC}_{b,1}$. By \cite[Theorem 4.2]{collins2015random}, the following holds almost surely as $d \to \infty$: if $\operatorname{supp}(\mathrm{SC}_{b,1}^{\boxplus n/k}) \subset (0,\infty)$, the map $\Phi_d^*$ (and thus $\Phi_d$) is asymptotically $k$-positive. Since $\mathrm{SC}_{b,1}^{\boxplus n/k} = \mathrm{SC}_{nb/k,\sqrt{n/k}}$, this condition is equivalent to 
\begin{equation}\label{eq:SN-PPT-b-k}
\frac{nb}{k} - 2 \sqrt{\frac n k} > 0.
\end{equation}

Let us now find a sufficient condition for $[\Phi_d \otimes \operatorname{id}_d](X_d) \ngeq 0$. Denoting by $\Omega_d$ the maximally entangled state in $\mathbb C^d \otimes \mathbb C^d$ and setting $\omega_d = \Omega_d \Omega_d^*$, we have
\begin{align*}
\langle  \Omega_d, [\Phi_d \otimes \operatorname{id}_d](X_d) \Omega_d \rangle &= \langle [\Phi_d^* \otimes \operatorname{id}_d](\omega_d), X_d \rangle\\
&= \langle C_d, X_d \rangle\\
&= \langle bI_{nd} + G_d, aI_{nd} -  G_d \rangle\\
&= (d+o(d))(ab-1),
\end{align*}
where we have used that the GUE matrix $G_d$ satisfies, almost surely,
$$\lim_{d \to \infty} \frac 1 d \operatorname{Tr} G_d = 0 \quad \text {and} \quad \lim_{d \to \infty} \frac 1 d \operatorname{Tr}(G_d^2) = 1.$$
Hence, if 
\begin{equation}\label{eq:SN-PPT-a-b-sigma}
ab - 1 < 0
\end{equation}
the matrix $[\Phi_d \otimes \operatorname{id}_d](X_d)$ is, asymptotically, not positive semidefinite. 

We claim that if the system of equations \eqref{eq:SN-PPT-a-sigma}, \eqref{eq:SN-PPT-b-k}, \eqref{eq:SN-PPT-a-b-sigma} has a solution in $a,b$, then the matrix $X_d$ is asymptotically PPT and $\operatorname{SN}(X_d) > k$. Indeed, the claim about the Schmidt number follows from the fact that the map $\Phi_d$ is asymptotically $k$-positive and, when applied to the $n$-part of $X_d$, it yields an output which is not positive semidefinite. Simple algebra shows that the system of 3 equations \eqref{eq:SN-PPT-a-sigma}, \eqref{eq:SN-PPT-b-k}, \eqref{eq:SN-PPT-a-b-sigma} has a solution in $a,b$ iff $k < n/16$. Taking $k = \lfloor (n-1)/16 \rfloor$ proves the claim about the Schmidt number and finishes the proof. 
 \end{proof}

\bibliography{../library}{}
\bibliographystyle{alpha}
\end{document}